\documentclass[a4paper,11pt]{article}
\pdfoutput=1
\usepackage[utf8]{inputenc}
\usepackage[T1]{fontenc}
\usepackage[UKenglish]{isodate}
\usepackage{mathrsfs}
\usepackage{bm}
\usepackage{bbm}
\usepackage{palatino}
\usepackage[margin=3cm]{geometry}
\usepackage{xurl}
\usepackage[usenames,dvipsnames]{xcolor}
\usepackage{amsmath}
\usepackage{amssymb}
\usepackage{amsthm}  
\usepackage{ifthen}
\usepackage{mathabx}
\usepackage{stmaryrd}
\usepackage{array}
\usepackage[colorlinks=true,urlcolor=Mahogany,linkcolor=Mahogany,citecolor=Mahogany,plainpages=false,pdfpagelabels, breaklinks]{hyperref}

\PassOptionsToPackage{linktocpage}{hyperref}

\usepackage{tikz}
\usepackage{verbatim}
\usetikzlibrary{shapes.geometric,plotmarks,backgrounds,fit,calc,circuits.ee.IEC}
\usetikzlibrary{decorations.pathreplacing}
\usepackage[backend=bibtex,maxcitenames=4,maxalphanames=100,maxbibnames=100,isbn=false, sorting = none]{biblatex}
\usepackage[tight,footnotesize]{subfigure}
\usepackage{cuted}
\newtheorem{theorem}{Theorem}
\newtheorem{lemma}[theorem]{Lemma}

\newtheorem{definition}[theorem]{Definition}

\newcommand{\sket}[1]{{\ensuremath{\lvert#1\rangle}}}
\newcommand{\lket}[1]{{\ensuremath{\left\lvert#1\right\rangle}}}
\newcommand{\ket}[1]{\mathchoice{\lket{#1}}{\sket{#1}}{\sket{#1}}{\sket{#1}}}

\newcommand{\sbra}[1]{{\ensuremath{\langle#1\rvert}}}
\newcommand{\lbra}[1]{{\ensuremath{\left\langle#1\right\rvert}}}
\newcommand{\bra}[1]{\mathchoice{\lbra{#1}}{\sbra{#1}}{\sbra{#1}}{\sbra{#1}}}

\DeclareFieldFormat{eprint:hal}{%
  HAL Id\addcolon\space
  \ifhyperref
      {\href{https://hal.inria.fr/hal-#1}{\texttt{#1}}}
    {\texttt{#1}}}

\DeclareFieldFormat{eprint:iacr}{%
  IACR eprint\addcolon\space
  \ifhyperref
    {\href{http://eprint.iacr.org/#1}{\texttt{#1}}}
    {\texttt{#1}}}

\DeclareFieldFormat{eprint:hal}{%
  HAL Id\addcolon\space
  \ifhyperref
      {\href{https://hal.archives-ouvertes.fr/#1}{\texttt{#1}}}
    {\texttt{#1}}}

\newcommand{\ident}{\mathbbm{1}}
\DeclareMathOperator{\tr}{Tr}

\newcommand{\mbE}{\mathbb{E}}

\renewcommand{\otimes}{\varotimes}
\newcommand{\htwo}{\tilde{H}^{\downarrow}_2}
\newcommand{\hhalf}{H^{\uparrow}_{\frac{1}{2}}}
\newcommand{\eqdef}{:=}

\usepackage{enumitem}

\begin{document}
\cleanlookdateon
\title{Quantum Polarization of Qudit Channels}
\author{Ashutosh Goswami${}^{1}$ \qquad Mehdi Mhalla${}^{2}$ \qquad Valentin Savin${}^3$\\[2mm]
{\small ${}^1$ Univ. Grenoble Alpes, Grenoble INP, LIG, F-38000 Grenoble, France}\\ %\\ \email{hoyer@ucalgary.ca}
{\small ${}^2$ Univ. Grenoble Alpes, CNRS, Grenoble INP, LIG, F-38000 Grenoble, France}\\%* Institute of Engineering Univ. Grenoble Alpes (note de bas de page, si possible %\\  \
{\small ${}^3$ Univ. Grenoble Alpes, CEA, LETI, F-38054 Grenoble, France}
}

\date{}

\maketitle

\begin{abstract} 
We provide a generalization of quantum polar codes to quantum channels with qudit-input, achieving the symmetric coherent information of the channel. Our \break scheme relies on  a channel combining and splitting construction, where a two-qudit unitary randomly chosen from a unitary 2-design is used to combine two instances of a qudit-input channel. The inputs to the synthesized bad channels are frozen by sharing EPR pairs between the sender and the receiver, so our scheme is entanglement assisted. Using the fact that the generalized two-qudit Clifford group forms a unitary 2-design, we conclude that the channel combining operation can be chosen from this set. Moreover, we show that polarization also happens for a much smaller subset of two-qudit Cliffords, which is not a unitary 2-design. Finally, we show how to decode the proposed quantum polar codes on Pauli qudit channels.
\end{abstract}

%We further show that the generalized two-qudit Clifford group forms a unitary 2-design, therefore the channel combining operation can be chosen from this set.
\section{Introduction}

In classical information theory, polar codes are the first explicit construction provably achieving the symmetric capacity of any discrete memoryless channel \cite{arikan09}. The construction is based on the recursive application of a channel combining and splitting procedure. It first combines two instances of the transmission channel, using a controlled-NOT  gate as channel combiner, and then splits the combined channel into two virtual channels, referred to as good and bad channels. Applied recursively $n$ times, the above procedure yields $N= 2^n$  virtual channels. These virtual channels exhibit a polarization property, in the sense that they tend to become either completely noisy or noiseless, as $N$ goes to infinity. %Moreover, the fraction of noiseless channels converges to the symmetric capacity of the transmission channel. 
Polar coding consists of efficient encoding and decoding algorithms that take effective advantage of the channel polarization property. 

Polar codes have been generalized to classical-quantum channels with binary and non-binary classical input in~\cite{wg13-2, nr18}. For the transmission of quantum information over quantum channels with qubit-input, two  approaches have been considered in the literature. The first approach is based on CSS-like constructions, which essentially exploit polarization in either amplitude or phase basis~\cite{rdr11,wg13,rw12}. The second approach relies on a  {\em purely quantum polarization} construction~\cite{our-itw-paper19, DGMS19}, where the synthesized virtual channels tend to become either completely noisy or noiseless as quantum channels, not merely in one basis. This approach uses a randomized channel combining, employing a random two-qubit Clifford unitary as channel combiner.

In this work, we extend the work in~\cite{our-itw-paper19} to the case of quantum channels with qudit-input. To the best of our knowledge, this is the first generalization of polar codes to qudit-input channels. First, we show that purely quantum polarization (in the sense of~\cite{our-itw-paper19}) happens for any qudit-input quantum channel, using as channel combiner a random two-qudit unitary, chosen from a unitary 2-design. Further, we provide a simple proof of the fact that the generalized two-qudit Clifford group forms a unitary 2-design, therefore the channel combining operation can be randomly chosen from this set. Moreover, when the qudit dimension $d$ is a prime, we show that polarization happens for a subset of two-qudit Clifford unitaries containing only $d^4 + d^2 - 2$ elements, which is not a unitary 2-design. Hence, unitary 2-designs are not necessary for the quantum polarization of qudit-input  channels. 

%Further, we show that the generalized two-qudit Clifford group forms a unitary 2-design, therefore the channel combining operation can be randomly chosen from this set.

To exploit the above polarization property, the inputs to the synthesized noisy channels are frozen by presharing EPR pairs between the sender and the receiver. Hence, our polar coding scheme is entanglement assisted. Finally, we consider the case of Pauli qudit channels. Similarly to~\cite{our-itw-paper19}, we associate a {\em classical counterpart channel} to a Pauli qudit channel. Then, we show that a quantum polar code on a Pauli qudit channel yields a classical polar code on the classical counterpart channel. Hence, we show that Pauli errors can be identified by decoding the polar code on the classical counterpart channel, using classical polar decoding.

The paper is organized as follows. Section~\ref{sec:prelims} provides the basic definitions needed for quantum polarization. Section~\ref{sec:quantum-polarization} contains our main polarization results for qudit-input quantum channels. Section~\ref{sec:qudit_Pauli} discusses the decoding of our quantum polar codes on Pauli qudit channels. 
%{\color{red} We shall sometimes refer to an extended, electronic preprint version of this paper~\cite{isit-paper-arxiv-version},  containing the proofs of some lemmas stated here.}

\section{Preliminaries}
\label{sec:prelims}

We consider $d$-dimensional quantum systems, referred to as qudits, where $d\geq 2$ is fixed throughout the paper. We  denote by $\rho_A$ a quantum state ({\em i.e.}, density matrix) of a quantum system $A$. When no confusion is possible, we shall discard the quantum system from the notation. For a bipartite quantum state $\rho_{AB}$, we shall denote by $\rho_B := \tr_A(\rho_{AB})$  the quantum state of the system $B$, obtained by tracing out the system $A$. The identity matrix is denoted by either $\ident$ or $I$, with the former notation  used for quantum states, and the latter for quantum operators. Throughout the paper, logarithm is taken in base $d$. 
\begin{definition}[von Neumann entropy]
\label{def:von-Neumann-entropy}
(a) The von Neumann entropy of a quantum state $\rho$ is defined as
$$H(\rho) := -\tr\left(\rho \log \rho \right).$$
(b) The conditional von Neumann entropy of a bipartite quantum state $\rho_{AB}$ is defined as
$$H(A|B)_{\rho_{AB}} = H(\rho_{AB}) - H(\rho_B).$$
\end{definition}
\begin{definition}[Conditional sandwiched R\'enyi entropy of order 2]
\label{def:Rényi-2-entropy}
    Let $\rho_{AB}$ be a quantum state. Then,
    $$\htwo(A|B)_{\rho} := -\log \tr\left[ \rho_B^{-\frac{1}{2}} \rho_{AB} \rho_B^{-\frac{1}{2}} \rho_{AB} \right].$$
\end{definition}
\begin{definition}[Petz-R\'enyi entropy of order $\frac{1}{2}$]
\label{def:Rényi-half-entropy}
    Let $\rho_{AB}$ be a quantum state. Then,
    $$\hhalf(A|B)_{\rho} := 2 \log \sup_{\sigma_B} \tr\left[ \rho_{AB}^{\frac{1}{2}} \sigma^{\frac{1}{2}}_B \right],$$
    where the supremum is taken over all quantum states $\sigma_B$. 
\end{definition}
%As shown in~\cite[Theorem 2]{tomamichel2014relating}, those two quantities satisfy a duality relation: given a pure tripartite state $\rho_{ABE}$, $\htwo(A|B)_{\rho} = -\hhalf(A|E)_{\rho}$.

\medskip We consider quantum channels $\mathcal{W}_{A' \rightarrow B}$, with qudit input system $A'$, and output system $B$ of arbitrary dimension. When no confusion is possible, we shall discard the channel input and output systems from the notation. An EPR pair on two-qudit systems $A$ and $A'$ is the quantum state $\Phi_{AA'} := \ket{\Phi_{AA'}}\bra{\Phi_{AA'}}$, with $\ket{\Phi_{AA'}} := \frac{1}{\sqrt{d}}\sum_{i=0}^{d-1}\ket{i}_A\ket{i}_{A'}$. Given a quantum channel $\mathcal{W}_{A' \rightarrow B}$, we denote by $\mathcal{W}(\Phi_{AA'}) := (I_A\otimes \mathcal{W})(\Phi_{AA'})$ the quantum state on the $AB$ system obtained by applying $\mathcal{W}$ on the $A'$-half of the EPR pair $\Phi_{AA'}$.

\begin{definition}[Symmetric coherent information]
\label{def:symmetric-coherent-information}
    Let $\mathcal{W}_{A' \rightarrow B}$ be a channel with qudit input $A'$ and output system $B$ of arbitrary dimension. The symmetric coherent information of  $\mathcal{W}$ is defined as the coherent information of the channel for a uniformly distributed input, that is
    $$I(\mathcal{W}) := -H(A|B)_{\mathcal{W}(\Phi_{AA'})} \in [-1, 1].$$ 
\end{definition}

We further introduce the following parameter of a quantum channel, which can be seen as the quantum counterpart of the classical Bhattacharyya parameter~\cite{our-itw-paper19}, and which we refer to as the ``R\'enyi-Bhattacharyya'' parameter. 

\begin{definition}[R\'enyi-Bhattacharyya parameter]
\label{def:Rényi_bhat}
Let $\mathcal{W}_{A' \rightarrow B}$ be a channel with qudit input $A'$ and output system $B$ of arbitrary dimension. Then,
$$
R(\mathcal{W}) := d^{\hhalf(A|B)_{\mathcal{W}(\Phi_{AA'})}} = d^{-\htwo(A|E)_{\mathcal{W}^c(\Phi_{AA'})}} \in \left[\tfrac{1}{d}, d\right],$$
where $\mathcal{W}^c$ denotes the complementary channel associated with $\mathcal{W}$~\cite{markwildebook}, and the equality \break $\hhalf(A|B)_{\mathcal{W}(\Phi_{AA'})} = -\htwo(A|E)_{\mathcal{W}^c(\Phi_{AA'})}$ follows from~\cite[Theorem 2]{tbh14}.
%where  $\mathcal{W}(\Phi_{A'A})$ is the quantum state on the $AB$ system obtained by applying $\mathcal{W}$ on the $A'$-half of the EPR pair $\Phi_{A'A}$
\end{definition}

We will also need the definitions of the generalized (qudit) Pauli and Clifford groups~\cite{Gottesman98, vlad14}, and  unitary $2$-designs~\cite{dcel09}. 
\begin{definition}[Generalized Pauli Group] 
(a) The Pauli operators $X$ and $Z$ for a qudit quantum system are defined as $X = \sum_{j=0}^{d-1} \ket{j}\bra{j \oplus 1}$, and $Z = \sum_{j=0}^{d-1} \omega^j \ket{j}\bra{j}$, where $\oplus$ denotes the sum modulo $d$, and $\omega = e^{\frac{2\pi \imath}{d}}$. 

\medskip\noindent (b) The generalized Pauli group on one qudit is defined as $\mathcal{P}_d^1 \eqdef \{ \omega^\lambda P_{r,s} \mid \lambda, r, s  = 0, \dots, d-1\}$, where $P_{r,s} \eqdef   X^r Z^s$.

\medskip\noindent (c) The generalized Pauli group on $n$ qudits is defined as $\mathcal{P}_d^n \eqdef \mathcal{P}_d^1  \otimes \mathcal{P}_d^1 \otimes  \cdots \otimes \mathcal{P}_d^1$. 
\end{definition}
It is easily seen that $X^d = Z^d = I$ and $XZ = \omega ZX$, hence $\mathcal{P}_d^1$ is indeed a group. Applying the commutation relation $XZ = \omega ZX$ appropriately many times, we have that 
\begin{equation} \label{eq:commutation-general}
P_{r,s} P_{t,u} = \omega^{ru - st} P_{t,u} P_{r,s}.
\end{equation}

\begin{definition}[Generalized Clifford Group] The Clifford group $\mathcal{C}_d^n$ is the unitary group on $n$ qudits that takes $\mathcal{P}_d^n$ to $\mathcal{P}_d^n$ by conjugation.
\end{definition}
 
Let $\mathcal{U}(d^n)$ be the set of unitary operators on $n$ qudits, and $ \mathcal{W}_n$ be a quantum channel with $n$-qudit input. The twirling of $\mathcal{W}_n$ with respect to $\mathcal{U}(d^n)$ is defined as the quantum channel that maps a $n$-qudit quantum state $\rho$ to $ \int  U^\dagger \mathcal{W}_n(U \rho U^\dagger) U d \eta $, where $U \in \mathcal{U}(d^n)$ is randomly chosen according to the Haar measure $\eta$.  The twirling of $\mathcal{W}_n$ with respect to a finite subset $\mathcal{U} \subset \mathcal{U}(d^n)$ is defined as the quantum channel acting as $\rho \mapsto \frac{1}{|\mathcal{U}|} \sum_{U \in \mathcal{U}}   U^\dagger \mathcal{W}_n (U \rho U^\dagger) U$.

\begin{definition}[Unitary 2-Design] \label{def:unitary-2-design}

 A finite subset $\mathcal{U} \subset \mathcal{U}(d^n)$ is said to form a unitary 2-design if it satisfies the following, for all $n$-qudit input quantum channels $\mathcal{W}_n$, and all $n$-qudit quantum states $\rho$:
\begin{equation} \label{eq:unitary_2design}
\frac{1}{|\mathcal{U}|} \sum_{U \in \mathcal{U}}   U^\dagger \mathcal{W}_n (U \rho U^\dagger) U = \int  U^\dagger \mathcal{W}_n(U \rho U^\dagger) U d \eta.
\end{equation}
\end{definition}

\section{Quantum Polarization of Qudit Channels}
\label{sec:quantum-polarization}

\subsection{Main polarization results}
Throughout this section ${\cal W}_{A' \rightarrow B}$ denotes a quantum channel with qudit input, and arbitrary dimension output.
Our quantum polarization scheme is based on the channel combining and splitting operations depicted in the following figure.  

\bigbreak 
\begin{figure*}[!h]
\centering
\subfigure[Combined channel]{%
\centering
    \begin{tikzpicture}[scale=0.5]
        \draw 
            (0,0) node[draw, minimum size=3mm] (canal1) {$\mathcal{W}$}
            (canal1) ++(0, -1.5) node[draw, minimum size=3mm] (canal2) {$\mathcal{W}$}
            ($.5*(canal1)+.5*(canal2)$) ++(-2.2, 0) node[draw, minimum height=1.5cm] (C) {$C$}
            ;

        \draw
            (canal1 -| C.east) to node[above] {} (canal1)
            (canal2 -| C.east) to node[above] {} (canal2)
            (canal1 -| C.west) to ++(-0.75, 0) node[left] {$A'_1$}
            (canal2 -| C.west) to ++(-0.75, 0) node[left] {$A'_2$}
            (canal1.east) to ++(0.75, 0) node[right] {$B_1$}
            (canal2.east) to ++(0.75, 0) node[right] {$B_2$}
            ;
            
        \node[draw, dashed,minimum height = 22mm, fit=(C) (canal1) (canal2)] (combined) {};
    \end{tikzpicture}
    \label{fig:combined-channel}
}
\subfigure[Bad channel $\mathcal{W}_{C}^{(0)}$]{%
\centering
    \begin{tikzpicture}[scale=0.5]   
        \draw 
            (0,0) node[draw, minimum size=3mm] (canal1) {$\mathcal{W}$}
            (canal1) ++(0, -1.5) node[draw, minimum size=3mm] (canal2) {$\mathcal{W}$}
            ($.5*(canal1)+.5*(canal2)$) ++(-2.2, 0) node[draw, minimum height=1.5cm] (C) {$C$}
            ;

        \draw
            (canal1 -| C.east) to node[above] {} (canal1)
            (canal2 -| C.east) to node[above] {} (canal2)
            (canal1 -| C.west) to ++(-3.0, 0) node[left] {$A'_1$}
            (canal2 -| C.west) to ++(-0.5, 0) node[left] (ground) {$\frac{\ident_{A'_2}}{d}$}
            (canal1.east) to ++(0.75, 0) node[right] (y1) {$B_1$}
            (canal2.east) to ++(0.75, 0) node[right] (y2) {$B_2$}
            ;

        \node[draw, dashed,fit=(C) (canal2) (ground)] { };
    \end{tikzpicture}
    \label{fig:bad-channel}
}\hfill%
\subfigure[Good channel $\mathcal{W}_{C}^{(1)}$]{%
\centering
    \begin{tikzpicture}[scale=0.5]
        \draw 
            (0,0) node[draw, minimum size=3mm] (canal1) {$\mathcal{W}$}
            (canal1) ++(0, -1.5) node[draw, minimum size=3mm] (canal2) {$\mathcal{W}$}
            ($.5*(canal1)+.5*(canal2)$) ++(-2.2, 0) node[draw, minimum height=1.5cm] (C) {$C$}
            ;

        \draw
            (canal1 -| C.east) to node[above] {} (canal1)
            (canal2 -| C.east) to node[above] {} (canal2)
            (canal2 -| C.west) to ++(-4.2, 0) node[left] {$A'_2$} 
            (canal1.east) to ++(0.75, 0) node[right] (y1) {$B_1$}
            (canal2.east) to ++(0.75, 0) node[right] {$B_2$}
            (canal1 -| C.west) to ++(-0.5, 0) to ++(-.5, .5) node[left] (EPR) {$\Phi_{A_1A'_1}$} coordinate (coude) to ++(.5, 1.0) coordinate (topline) to (topline -| y1.west) node[right] (R) {$A_1$}
            ;

        \node[draw, dashed, fit=(EPR) (C) (coude) (topline) (canal2)] { };
    \end{tikzpicture}
    \label{fig:good-channel}
}   
\caption{Channel combining and splitting. (a) combined channel: a two-qudit unitary $C$ is applied on the two inputs. (b) bad channel: we input a totally mixed state into the second input. (c) good channel: we input half of an EPR pair into the first input, and the other half becomes the output $A_1$.}
    \label{fig:channel-combining-splitting}
\end{figure*}
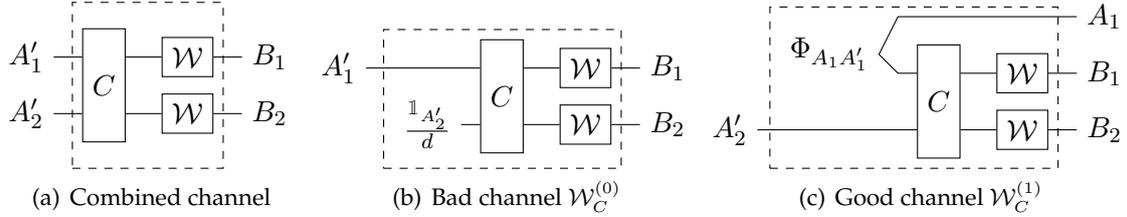
%\vspace{-.2\baselineskip}

\medskip First, two  instances of ${\cal W}$ are combined, by entangling their inputs through a two-qudit unitary $C$. The combined channel is then split into one  bad and one good channel. The bad channel $\mathcal{W}_{C}^{(0)}$ is a channel from $A'_1$ to $B_1 B_2$ that acts as $\mathcal{W}_{C}^{(0)}(\rho) \break = \mathcal{W}^{\otimes 2}\left( C (\rho \otimes \frac{\ident_{A'_2}}{d}) C^{\dagger} \right)$, where $\frac{\ident_{A'_2}}{d}$ is the completely mixed state. The good channel $\mathcal{W}_{C}^{(1)}$ is a channel from $A'_2$ to $A_1 B_1 B_2$ that acts as $\mathcal{W}_{C}^{(1)}(\rho) = \mathcal{W}^{\otimes 2}\left( C(\Phi_{A_1 A'_1} \otimes \rho) C^{\dagger} \right)$, where $\Phi_{A_1 A'_1}$ is an EPR pair.

The polarization construction is obtained by recursively applying the above channel combining and splitting operations, while choosing $C$ randomly from some finite set of unitaries, denoted by ${\cal U} \subset {\cal U}(d^2)$. To accommodate the random choice of $C\in {\cal U}$,  a classical description of $C$ is included as part of the output of the bad and good channels.  Hence, for $i=0,1$, we define:
\begin{equation}\label{eq:c-as-output}
{\cal W}^{(i)}(\rho) = \frac{1}{|{\cal U}|} \sum_{C \in {\cal U}} \ket{C}\bra{C} \otimes {\cal W}_C^{(i)}(\rho),
\end{equation}
where  $\{ \ket{C} \}_{C \in {\cal U}}$ is an orthogonal basis of some auxiliary system. Applying twice the  transformation ${\cal W} \mapsto \left({\cal W}^{(0)},{\cal W}^{(1)}\right)$, we get channels ${\cal W}^{(i_1 i_2)} \eqdef \left( {\cal W}^{(i_1)} \right)\,\!^{(i_{2})}$,  where  $(i_1i_{2}) \in \{00, 01, 10, 11\}$. In general, after $n$ levels or recursion, we obtain $2^n$ channels:
\begin{equation}\label{eq:recursive_construction}
{\cal W}^{(i_1\dots i_{n})} \eqdef \left( {\cal W}^{(i_1\dots i_{n-1})} \right)\,\!^{(i_{n})}, \ \forall  (i_1\dots i_{n}) \in \{0,1\}^{n}.  
\end{equation}

The quantum polarization theorem below states that the symmetric coherent information of the synthesized channels ${\cal W}^{(i_1\dots i_{n})}$ polarizes, meaning that it goes to either $-1$ or $+1$ as $n$ goes to infinity (except possibly for a vanishing fraction of channels), provided that ${\cal U}$ is a unitary $2$-design. The second theorem states that polarization also happens when ${\cal U}$ is taken to be the generalized Clifford group on two qudits,  ${\cal C}_d^2$, or some specific subset of it.

\begin{theorem}\label{thm:quantum_polarization}
Let $\mathcal{U}$ be a unitary 2-design. For any qudit-input quantum channel ${\cal W}$, let  $ \break \left\{ {\cal W}^{(i_1\dots i_{n})}  : (i_1\dots i_{n}) \in \{0,1\}^n \right\}$ be the set of channels defined in~(\ref{eq:recursive_construction}), with channel combining unitary $C$ randomly chosen from ${\cal U}$. Then, for any $\delta > 0$, 
$$\displaystyle\lim_{n\rightarrow\infty} \frac{\#\{(i_1\dots i_{n}) \in \{0,1\}^{n} : I\left( {\cal W}^{(i_1\dots i_{n})} \right) \in (-1+\delta, 1-\delta)  \}}{2^n} = 0$$
\noindent and furthermore,
$$\displaystyle \lim_{n \rightarrow \infty} \frac{\#\left\{ (i_1,\dots,i_n) \in \{0,1\}^n : I(\mathcal{W}^{(i_1,\dots,i_n)}) \geqslant 1-\delta\right\} }{2^n} = \frac{I(\mathcal{W}) + 1}{2}$$
\end{theorem}

\begin{theorem}\label{thm:clifford-is-2-design}
(a) The generalized Clifford group on two qudits, ${\cal C}_d^2$, is a unitary $2$-design. Thus, polarization happens when the channel combining unitary $C$ is randomly chosen from  ${\cal C}_d^2$.

\medskip\noindent (b) If $d$ is prime, there exists a subset ${\cal U} \subset {\cal C}_d^2$, of size $|{\cal U}| = d^4+d^2-2$, which is not a unitary $2$-design, and such that polarization happens when the channel combining unitary $C$ is randomly chosen from ${\cal U}$.
\end{theorem}

We note that part (a) of Theorem~\ref{thm:clifford-is-2-design} may be inferred from Lemmas 1, 2 and 3 in \cite{webb16}. We will give an alternative and more elementary proof in Section~\ref{sec:clifford-is-2-design}, by generalizing the proof from \cite{dcel09} to the qudit case.

\subsection{Proof of Theorem~\ref{thm:quantum_polarization} (quantum polarization)}

To prove the polarization theorem, we essentially need three ingredients, as follows.
\begin{enumerate}[leftmargin=1.3\parindent]
\item For any two-qudit unitary $C$, the total symmetric coherent information is preserved under channel combining and splitting, that is, $I({\cal W}_C^{(0)}) + I({\cal W}_C^{(1)}) = 2I({\cal W})$. We omit the proof of this, as the proof given in~\cite[Lemma 10]{DGMS19} for qubit-input channels remains valid in the qudit case, with minor adjustments.

\item The symmetric coherent information $I({\cal W})$ approaches $\{-1, +1\}$ values if and only if the R\'enyi-Bhattacharyyia parameter $R({\cal W})$ approaches $\{d, 1/d\}$ values. This follows from Lemma~\ref{lem:i-and-t-1}, below. 

\item Taking the good channel yields a guaranteed improvement of the average R\'enyi-Bhattacharyya parameter, in the sense of Lemma~\ref{lem:d-good-channel}, below.
\end{enumerate}

\smallskip \noindent The proof of Theorem~\ref{thm:quantum_polarization} then follows by using~\cite[Lemma 7]{DGMS19}, similar to the proof of quantum polarization for qubit-input channels in~\cite{DGMS19}.

%\begin{figure*}[!t]
%\begin{align}
%\mathcal{W}_2'(\rho) &:= \frac{1}{d^5}\sum_{\lambda,r,s,r',s'} \left( \omega^\lambda P_{r,s} \otimes P_{r',s'}\right)^\dagger A \left( \omega^\lambda P_{r,s} \otimes P_{r',s'}\right) \rho \left( \omega^\lambda P_{r,s} \otimes P_{r',s'}\right)^\dagger B \left( \omega^\lambda P_{r,s} \otimes P_{r',s'}\right), \label{eq:twirling-P2}\\
%%
%&= \frac{1}{d^4}\sum_{r,s,r',s'} ( P_{r,s}^\dagger \otimes P_{r',s'}^\dagger) A \left( P_{r,s} \otimes P_{r',s'}\right) \rho ( P_{r,s}^\dagger \otimes P_{r',s'}^\dagger) B \left( P_{r,s} \otimes P_{r',s'}\right). \label{eq:twirling-barP2} \\
%%
%\mathcal{W}_2'(\rho) &= \sum_{r,s,r',s'} \gamma_{r,s,r',s'} \left(P_{r,s} \otimes P_{r',s'}\right) \rho (P_{r,s}^\dagger \otimes P_{r',s'}^\dagger),\ \text{ where }
%\gamma_{r,s,r',s'} := \omega^{rs+r's'} \alpha(r,s,r',s') \beta(-r,-s,-r',-s'),
%\label{eq:twirling-P2-is-Pauli} \\[-3mm]
% & \hspace*{75mm} \text{and} -x \text{ denotes the additive inverse of } x \text{ modulo } d. \nonumber \\[2mm] 
%%
%\mathcal{W}_2''(\rho) &= \frac{\tr(AB)}{d^4} \ident\otimes\ident + \frac{d^2\tr(A)\tr(B) - \tr(AB)}{d^2(d^4-1)} \left( \rho - \frac{1}{d^2} \ident\otimes\ident\right).
%\label{eq:twirl-C2}
%\end{align}
%\hrulefill
%\end{figure*}
%\color{red}\tr(\rho)

\begin{lemma}\label{lem:i-and-t-1}
    Let $\mathcal{W}_{A' \rightarrow B}$ be a channel with qudit input. Then,
    \begin{enumerate}
        \item $R(\mathcal{W}) \leqslant \frac{1}{d}+\delta \Rightarrow I(\mathcal{W}) \geqslant 1 - \log(1+d\delta)$.
        \item $R(\mathcal{W}) \geqslant d - \delta \Rightarrow  I(\mathcal{W}) \leqslant -1 + 2\sqrt{\frac{\delta}{d}} + \frac{\sqrt{d}+\sqrt{\delta}}{\sqrt{d}} h\left( \frac{\sqrt{\delta}}{\sqrt{d} + \sqrt{\delta}} \right)$, where $h(\cdot)$ denotes the binary entropy function. 
    \end{enumerate}
\end{lemma}
\begin{proof}
       We prove first 1). For $\rho_{AB} = \mathcal{W}(\Phi_{AA'})$, we have that
    $$ \frac{1}{d} + \delta \geqslant R(\mathcal{W})
        = d^{\hhalf(A|B)_{\rho}}
        \geqslant d^{H(A|B)_{\rho}}
        = d^{-I(\mathcal{W})},$$
 where we have used $\hhalf(A|B)_{\rho} \geqslant H(A|B)_{\rho}$ for the second inequality, which follows from the monotonically decreasing property of the conditional Petz-Rényi entropy with respect to its order \cite[Theorem 7]{mdsft13}. Hence, $I(\mathcal{W}) \geqslant 1 - \log(1+d\delta)$. 

%{\color{red}    
    \smallskip We now turn to point 2). We have that 
    \begin{align}
    d - \delta &\leqslant R(\mathcal{W}) \leqslant R(\mathcal{W}) \nonumber \\
    &= \max_{\sigma_B} \tr\left[ \rho^{\frac{1}{2}}_{AB} \sigma^{\frac{1}{2}}_B \right]^2 \nonumber \\
    & = d \max_{\sigma_B} \tr\left[ \sqrt{\rho_{AB}} \sqrt{\frac{\ident_A}{d} \otimes \sigma_B} \right]^2 \nonumber \\
    & \leqslant d \max_{\sigma_B} \left\| \sqrt{\rho_{AB}} \sqrt{\frac{\ident_A}{d} \otimes \sigma_B} \right\|_1^2 \\
    &= d \max_{\sigma_B} F\left( \rho_{AB}, \frac{\ident_A}{d} \otimes \sigma_B \right)^2
    \end{align}
%    \noindent $\displaystyle d - \delta \leqslant R(\mathcal{W}) $\hfill\,
%    
%    \hfill$\displaystyle = \max_{\sigma_B} \tr\left[ \rho^{\frac{1}{2}}_{AB} \sigma^{\frac{1}{2}}_B \right]^2
%        = d \max_{\sigma_B} \tr\left[ \sqrt{\rho_{AB}} \sqrt{\frac{\ident_A}{d} \otimes \sigma_B} \right]^2$
%        
%     \noindent\resizebox{\linewidth}{!}{ $\displaystyle \leqslant d \max_{\sigma_B} \left\| \sqrt{\rho_{AB}} \sqrt{\frac{\ident_A}{d} \otimes \sigma_B} \right\|_1^2
%        = d \max_{\sigma_B} F\left( \rho_{AB}, \frac{\ident_A}{d} \otimes \sigma_B \right)^2$}
%    
 Using the Fuchs-van de Graaf inequalities~\cite{fuchs-vdg}, we get that there exists a $\sigma_B$ such that
    $ \break \frac{1}{2} \left\| \rho_{AB} - \frac{\ident_A}{d} \otimes \sigma_B \right\|_1 \leqslant \sqrt{\frac{ \delta}{d}}$.
    %\[ \left\| \rho_{AB} - \frac{\ident_A}{2} \otimes \sigma_B \right\|_1 \leqslant \sqrt{2\delta}. \]
    We are now in a position to use the Alicki-Fannes-Winter~\cite[Lemma 2]{winter16} inequality, which states that
    \begin{align*}
        \left| H(A|B)_{\rho} - 1 \right| \leqslant 2 \sqrt{\frac{\delta}{d}} + \frac{\sqrt{d} + \sqrt{\delta}}{\sqrt{d}} h\left(\frac{\sqrt{\delta}}{{\sqrt{d} + \sqrt{\delta}}}\right).
    \end{align*}
    This concludes the proof of the lemma.
\end{proof}

\begin{lemma}\label{lem:d-good-channel}
Let $\mathcal{W}_{A' \rightarrow B}$ be a channel with qudit input. Then, 
$$\mbE_C R\left(\mathcal{W}_{C}^{(1)}\right) = \frac{d}{d^2+1} \left(1 + R(\mathcal{W})^2 \right) \leq R(\mathcal{W}), $$
where $\mbE_C$ denotes the expectation operator, $C$ is the channel combining unitary, chosen uniformly at random from a unitary $2$-design ${\cal U}$. Moreover, equality happens if and only if $R(\mathcal{W}) \in \{1/d, d\}$.
\end{lemma}
\begin{proof}
   Let $\mathcal{W}^c_{A' \rightarrow E}$  and  $(\mathcal{W}_{C}^{(1)})_{{{A'_2 \to E_1E_2}}}^c$ be the complementary channel associated with $\mathcal{W}_{A' \rightarrow B}$ and the good channel $ \mathcal{W}^{(1)}_{C_{A'_2 \to A_1 B_1B_2}}$, respectively. The complementary of the good channel acts as $(\mathcal{W}_{C}^{(1)})^c(\rho) = (\mathcal{W}^c \otimes \mathcal{W}^c)\left( C \left( \frac{\ident_{A'_1}}{d} \otimes \rho \right) C^{\dagger} \right)$ (see ~\cite[Appendix A]{DGMS19} for a proof). Therefore, $R(\mathcal{W}_{C}^{(1)}) = d^{-\htwo(A_2|E_1 E_2)_{\rho}}$,
    %\[ R(\mathcal{W} \varoast \mathcal{M}) = 2^{-\htwo(A_2|E_1 E_2)_{\rho}}, \]
    where $\rho_{A_2 E_1 E_2} = (\mathcal{W}_{C}^{(1)})^c(\Phi_{A_2 A'_2})$. Note that $\rho_{E_1 E_2} = \mathcal{W}^c\left(\frac{\ident}{d} \right) \otimes \mathcal{W}^c\left(\frac{\ident}{d} \right)$, which is independent of $C$.  To compute the expected value of $R(\mathcal{W}_{C}^{(1)})$ with respect to $C$, we proceed as follows.
\begin{align}
\mbE_C d^{-\htwo(A_2|E_1 E_2)_{\rho}} &= \mbE_C \tr\left[ \left( \rho_{E_1 E_2}^{-\frac{1}{4}} \rho_{A_2 E_1 E_2} \rho_{E_1 E_2}^{-\frac{1}{4}} \right)^2 \right] \nonumber \\
&= \mbE_C \tr\left[ \left( \rho_{E_1 E_2}^{-\frac{1}{4}} (\mathcal{W}^c \otimes \mathcal{W}^c)\left( C \left( \frac{\ident_{A'_1}}{d} \otimes \Phi_{A_2 A'_2} \right) C^{\dagger} \right) \rho_{E_1 E_2}^{-\frac{1}{4}} \right)^2 \right]. \nonumber 
\end{align}   
%
%    \hfill $\mbE_C d^{-\htwo(A_2|E_1 E_2)_{\rho}} = \mbE_C \tr\left[ \left( \rho_{E_1 E_2}^{-\frac{1}{4}} \rho_{A_2 E_1 E_2} \rho_{E_1 E_2}^{-\frac{1}{4}} \right)^2 \right]$ \\
%    \resizebox{\linewidth}{!}{$= \mbE_C \tr\left[ \left( \rho_{E_1 E_2}^{-\frac{1}{4}} (\mathcal{W}^c \otimes \mathcal{W}^c)\left( C \left( \frac{\ident_{A'_1}}{d} \otimes \Phi_{A_2 A'_2} \right) C^{\dagger} \right) \rho_{E_1 E_2}^{-\frac{1}{4}} \right)^2 \right].$}
%%
%    \begin{align*}
%        \mbE_C 2^{-\htwo(A_2|E_1 E_2)_{\rho}} = \mbE_C \tr\left[ \left( \rho_{E_1 E_2}^{-\frac{1}{4}} \rho_{A_2 E_1 E_2} \rho_{E_1 E_2}^{-\frac{1}{4}} \right)^2 \right]\\
%        = \mbE_C \tr\left[ \left( \rho_{E_1 E_2}^{-\frac{1}{4}} (\mathcal{W}^c \otimes \mathcal{M}^c)\left( C \left( \frac{\ident_{A'_1}}{2} \otimes \Phi_{A_2 A'_2} \right) C^{\dagger} \right) \rho_{E_1 E_2}^{-\frac{1}{4}} \right)^2 \right].
%    \end{align*}
 Note that this is basically the same calculation as in~\cite[Equation~(3.32)]{fred-these} (there, $U$ is chosen according to the Haar measure over the full unitary group, but all that is required is a unitary 2-design). However, we will not make the simplifications after (3.44) and (3.45) in \cite{fred-these}, but will instead keep all the terms. We therefore get %the following result:\\
    $\mbE_C d^{-\htwo(A_2|E_1 E_2)_{\rho}} = \alpha \tr\left[ (\frac{\ident_{A_2}}{d})^2 \right] + \beta \tr\left[ (\frac{\ident_{A_1'}}{d} \otimes \Phi_{A_2 A'_2})^2 \right]
        = \frac{1}{d} \alpha + \frac{1}{d} \beta$,
%    \begin{align*}
%        \mbE_C 2^{-\htwo(A_2|E_1 E_2)_{\rho}} &= \alpha \tr\left[ \pi_{A_2}^2 \right] + \beta \tr\left[ \pi_{A'_1}^2 \otimes \Phi_{A_2 A'_2} \right]\\
%        &= \frac{1}{2} \alpha + \frac{1}{2} \beta,
%    \end{align*}
    where
    $\alpha  = \frac{d^4}{d^4 - 1} - \frac{d^2}{d^4 - 1} d^{-\htwo(A_1 A_2|E_1 E_2)_{\omega}}$, 
    $\beta = \frac{d^4}{d^4 - 1}  d^{-\htwo(A_1 A_2|E_1 E_2)_{\omega}}- \frac{d^2}{d^4 - 1}  $, and  $\omega_{A_1 A_2 E_1 E_2} := (\mathcal{W}^c \otimes \mathcal{W}^c)(\Phi_{A_1 A'_1} \otimes \Phi_{A_2 A'_2})$.  Hence,
    \begin{align*}
        \mbE_C d^{-\htwo(A_2|E_1 E_2)_{\rho}} &= \frac{d}{d^2 + 1} + \frac{d}{d^2 + 1} d^{-\htwo(A_1 A_2|E_1 E_2)_{\omega}}\\
        &= \frac{d}{d^2 + 1} (  1 +  R(\mathcal{W})^2 ),
    \end{align*}
where the second equality follows from $d^{-\htwo(A_1 A_2|E_1 E_2)_{\omega}} = R(\mathcal{W})^2$ using the fact that conditional sandwiched Rényi entropy of order 2 is additive with respect to tensor-product states. It is easily seen that the function  $f(R) = \frac{d}{d^2 + 1} (  1 +  R^2 )$ is a convex function satisfying $f(R) = R$ for $R \in \{ \frac{1}{d}, d \}$ and $f(R) < R$ for $R \in (\frac{1}{d}, d)$.
\end{proof}

\subsection{Proof of Theorem~\ref{thm:clifford-is-2-design}} \label{sec:clifford-is-2-design}

\noindent {\bf \em Proof of part (a).} It is shown in~\cite[Theorem 1]{dcel09} (see also~\cite{notes_Olivia}) that the Clifford group on $n$-qubits
forms a unitary $2$-design for any $n\geq 1$. Here, we generalize the proof from~\cite{dcel09} to the qudit case, and for $n=2$.  We need to prove that the Clifford group $\mathcal{C}_d^2$ satisfies the Definition~\ref{def:unitary-2-design}. For this, it is sufficient to prove (\ref{eq:unitary_2design}), with $\mathcal{U} = \mathcal{C}_d^2$, for two-qudit input quantum channels  of the form $\mathcal{W}_2(\rho) := A\rho B$ (since any quantum channel is a convex combination of quantum channels of this form).

 We first consider the twirling of $\mathcal{W}_2$ with respect to the Clifford group $\mathcal{C}_d^2$. Since the Pauli group $\mathcal{P}_d^2$ is a normal subgroup of  $\mathcal{C}_d^2$, we may chose a subset $\bar{\mathcal{C}}_d^2 \subset  \mathcal{C}_d^2$ containing one representative for each  equivalence class in the quotient group $\mathcal{C}_d^2 / \mathcal{P}_d^2$. Thus, any element of  $\mathcal{C}_d^2$ can be uniquely written as a product $CP$, where $C \in \bar{\mathcal{C}}_d^2$, and $P \in  \mathcal{P}_d^2$. Therfore, in order to twirl $\mathcal{W}_2$ with respect to $\mathcal{C}_d^2$, we may first twirl it with respect to $\mathcal{P}_d^2$, then twirl again the obtained channel with respect to  $\bar{\mathcal{C}}_d^2$.

\medskip The elements of $\mathcal{P}_d^2$ have the form $\omega^\lambda P_{r,s} \otimes P_{r',s'}$, with $\lambda,r,s,r',s' = 0,\dots, d-1$. Hence, twirling $\mathcal{W}_2$ with respect to $\mathcal{P}_d^2$ gives a quantum channel, denoted $\mathcal{W}_2'$, defined below
{\small
\begin{align}
\mathcal{W}_2'(\rho) &:= \frac{1}{d^5}\sum_{\lambda,r,s,r',s'} \left( \omega^\lambda P_{r,s} \otimes P_{r',s'}\right)^\dagger A \left( \omega^\lambda P_{r,s} \otimes P_{r',s'}\right) \rho \left( \omega^\lambda P_{r,s} \otimes P_{r',s'}\right)^\dagger B \left( \omega^\lambda P_{r,s} \otimes P_{r',s'}\right),  \nonumber \\
%\label{eq:twirling-P2}
&= \frac{1}{d^4}\sum_{r,s,r',s'} ( P_{r,s}^\dagger \otimes P_{r',s'}^\dagger) A \left( P_{r,s} \otimes P_{r',s'}\right) \rho ( P_{r,s}^\dagger \otimes P_{r',s'}^\dagger) B \left( P_{r,s} \otimes P_{r',s'}\right). \label{eq:twirling-barP2} 
\end{align} 
}
%in~(\ref{eq:twirling-barP2})
The last equality from the above shows that it is actually enough to twirl $\mathcal{W}_2$ with respect to the subset $\bar{\mathcal{P}}_d^2 := \left\{ P_{r,s} \otimes P_{r',s'} \mid  r,s,r',s' = 0,\dots, d-1 \right\}$, obtained by omitting phase factors. Since $\bar{\mathcal{P}}_d^2$ forms an operator basis (for two-qudit operators), we may write 
$A = \sum_{r,s,r',s'} \alpha(r,s,r',s')  P_{r,s} \otimes P_{r',s'}$, and
$B = \sum_{r,s,r',s'} \beta(r,s,r',s')  P_{r,s} \otimes P_{r',s'}$.
%
%\begin{align}
%A &= \sum_{r,s,r',s'} \alpha(r,s,r',s')  P_{r,s} \otimes P_{r',s'}\\
%B &= \sum_{r,s,r',s'} \beta(r,s,r',s')  P_{r,s} \otimes P_{r',s'}
%\end{align}
%
The following two lemmas are proven in Appendix~\ref{appx:proof:lem:twirling-Pd} and Appendix~\ref{appx:proof:lem:twirling-Cd}, respectively.
\begin{lemma} \label{lem:twirling-Pd}
The quantum channel $\mathcal{W}_2'$, obtained by twirling $\mathcal{W}_2$ with respect to $\bar{\mathcal{P}}_d^2$, is a Pauli channel satisfying the following
\begin{align}
\mathcal{W}_2'(\rho) &= \sum_{r,s,r',s'} \gamma_{r,s,r',s'} \left(P_{r,s} \otimes P_{r',s'}\right) \rho (P_{r,s}^\dagger \otimes P_{r',s'}^\dagger), \label{eq:twirling-P2-is-Pauli}
\end{align}
where $\gamma_{r,s,r',s'} := \omega^{rs+r's'} \alpha(r,s,r',s') \beta(-r,-s,-r',-s')$ and $-x$ denotes the additive inverse of $x$ modulo $d$.
\end{lemma}
\begin{lemma} \label{lem:twirling-Cd}
The quantum channel obtained by twirling $\mathcal{W}_2'$ with respect to $\bar{\mathcal{C}}_d^2$, is the quantum channel $\mathcal{W}_2''$ acting as
\begin{align}
\mathcal{W}_2''(\rho) = \frac{\tr(AB)}{d^4} \ident\otimes\ident + \frac{d^2\tr(A)\tr(B) - \tr(AB)}{d^2(d^4-1)} \left( \rho - \frac{1}{d^2} \ident\otimes\ident\right). \label{eq:twirl-C2}
\end{align}
\end{lemma}

Now, the quantum channel  $\mathcal{W}_2''$ from~(\ref{eq:twirl-C2}) is the twirling of $\mathcal{W}_2$ with respect to $\mathcal{C}_d^2$. To conclude that $\mathcal{C}_d^2$ is a unitary 2-design, we need to show that twirling  $\mathcal{W}_2$ with respect to $\mathcal{U}(d^2)$ yields the same channel, which follows from~\cite{eaz05}.

\medskip \noindent {\bf\em Proof of part (b).}   We  will need the following two lemmas. The first is basically the same as~\cite[Lemma 14]{DGMS19} and the proof can be easily generalized. The second is proven in Appendix~\ref{appx:proof:lem:number-Clifford}.
\begin{lemma}\label{lem:equivalence-Rényi-bhattachayya}
Consider $C, C' \in \mathcal{C}_d^2$, such that $C' = C (C_1 \otimes C_2)$, for some $C_1, C_2 \in  \mathcal{C}_d^1$. Then, $C' $ and $C''$ yield the same Rényi-Bhattacharya parameter for both good and bad channels, i.e., following equalities hold,
\begin{itemize}
\item[1)] $R(\mathcal{W}_C^{(0)}) =  R(\mathcal{W}_{C'}^{(0)}).$

\item[2)] $ R(\mathcal{W}_C^{(1)}) =  R(\mathcal{W}_{C'}^{(1)}).$
\end{itemize}
\end{lemma}
\begin{lemma} \label{lem:number-Clifford}
If $d$ is a prime number,  $|\mathcal{C}_d^1| = d^3(d^2 - 1)$ and $|\mathcal{C}_d^2| = d^8 (d^4 - 1) (d^2 - 1)$.
\end{lemma}

%\begin{lemma} \label{lem:SWAP-gate}
%For any $C \in \mathcal{C}_d^2$, there is no $C_1, C_2 \in \mathcal{C}_d^1 $ such that $C(C_1 \otimes C_2) = SC$, where $S$ is the SWAP gate.
%\end{lemma}
%{\color{blue} Here we also need to consider the second symmetry (swap gate), to divide the number of Cliffords by 2. We may also take off the identity gate, thus getting a subset of Cliffords $\mathcal{U} \subset \mathcal{C}_d^2$, of size $|\mathcal{U}| = \frac{d^4 + d^2}{2} - 1$ (as stated in the theorem). 
%
%After you include the content of the current version (+ swap gate update), you may conclude: 
%
%\medskip
%}

  We are now in a position to prove the part (b) of the theorem. The group $\mathcal{C}_d^2$ can be decomposed into left cosets with respect to the subgroup $\mathcal{C}_d^1 \otimes  \mathcal{C}_d^1 \subset \mathcal{C}_d^2$. From Lemma~\ref{lem:equivalence-Rényi-bhattachayya}, it follows that any two elements in the same left coset, when used as channel combiners, yield the same Rényi-Bhattacharyya parameter for both good and bad channels. Therefore, polarization also happens for any subset ${\cal L}\subset \mathcal{C}_d^2$, containing  one representative of each left coset (since $\mbE_{C\in {\cal L}} R(\mathcal{W}_{C}^{(1)}) = \mbE_{C\in \mathcal{C}_d^2} R(\mathcal{W}_{C}^{(1)})$, thus the guaranteed improvement of the average Rényi-Bhattacharyya parameter, in the  sense of  Lemma~\ref{lem:d-good-channel}, still holds when $C$ is randomly chosen from ${\cal L}$).
 % when we restrict $\mathcal{C}_d^2$ to any set containing only one representative from each left coset (as this restriction doesn't affect the expectation value  $\mbE_C R(\mathcal{W}_{C}^{(1)})$). 
 Using Lemma~\ref{lem:number-Clifford}, the number of cosets of $\mathcal{C}_d^1 \otimes  \mathcal{C}_d^1$ in $\mathcal{C}_d^2$ is equal to $\frac{|\mathcal{C}_d^2|}{|\mathcal{C}_d^1 \otimes  \mathcal{C}_d^1|} = d^4 + d^2$, therefore ${\cal L}$ contains $d^4 + d^2$ representatives, two of which may be chosen to be the identity ($I$) and the swap ($S$) operators. Since $R(\mathcal{W}_{I}^{(1)}) = R(\mathcal{W}_{S}^{(1)}) = R(\mathcal{W}) \geq \mbE_{C\in {\cal L}} R(\mathcal{W}_{C}^{(1)})$, we may further remove $I$ and $S$ from ${\cal L}$, thus getting a subset $\mathcal{L}' :=  \mathcal{L} \setminus \{I, S\}$ containing $d^4 + d^2 - 2$ elements, which still ensures polarization of qudit-input quantum channels. 
From~\cite{gkj07, rs09}, we know that a set of unitaries in dimension $\delta$ can only form a unitary 2-design if it has  at least $\delta^4 - 2\delta^2 + 2$ elements.  As we consider a two-qudit system (dimension $\delta = d^2$), a unitary 2-design would have at least $d^8 - 2d^4 + 2$ two-qudit unitaries, which is clearly  bigger than $d^4 + d^2 - 2$. Hence, the set $\mathcal{L}'$ is not a unitary $2$-design. This completes the poof of the part (b). 
\hfill \qedsymbol

\smallskip One may try to further reduce the size of ${\cal L}'$, by considering the action of the swap gate $S$. Indeed, it can be seen that the two equalities from Lemma~\ref{lem:equivalence-Rényi-bhattachayya} also hold for two $C, C' \in \mathcal{C}_d^2$, such that $C' = SC$ ({see also~\cite[Lemma 15]{DGMS19}}). Hence, if both $C$ and $C'$ belong to ${\cal L}'$, one of them can be removed, while still ensuring polarization. Now, multiplying by $S$ on the left induces a permutation on the left cosets of $\mathcal{C}_d^1 \otimes  \mathcal{C}_d^1$ in $\mathcal{C}_d^2$, which in turn induces a permutation ${\cal L}'\stackrel{\sim}{\rightarrow}{\cal L}'$. In the qubit case ($d=2$), this permutation has no fixed points, thus the size of ${\cal L}'$ can be reduced by half. However, in general the above permutation may have fixed points. We provide such an example in Appendix~\ref{appx:example-left-coset}, where we show that for $d=5$, there exist $C\in \mathcal{C}_d^2$ and $C_1, C_2 \in  \mathcal{C}_d^1$, such that $SC = C(C_1\otimes C_2)$.

% %%%%%%%%%%%%%%%%%%%%%%%%%%%%%%%%%%%%%%%%%%%%%%%%%%%% % 
%            Polarization of Pauli Channels
% %%%%%%%%%%%%%%%%%%%%%%%%%%%%%%%%%%%%%%%%%%%%%%%%%%%% %

\section{Quantum Polar codes on Pauli Qudit channels}\label{sec:qudit_Pauli}

In this section, we discuss the decoding of quantum polar codes on a Pauli qudit channel. We shall assume that all channel combining unitaries are Clifford unitaries.

A Pauli qudit channel $\mathcal{W}$ is defined as the quantum channel that maps a qudit quantum state $\rho$ to $\sum_{r,s} a_{r,s} P_{r,s} \rho P_{r,s}^\dagger$, where $a_{r,s} \geq 0$ with $\sum_{r,s} a_{r,s} =1$. Similar to~\cite[Definition 17]{DGMS19}, we  associate a classical channel with $\mathcal{W}$, which is referred to as the classical counterpart of $\mathcal{W}$, and denoted by $\mathcal{W}^{\#}$. The classical counterpart $\mathcal{W}^{\#}$ is a classical channel with input and output alphabet $\bar{\mathcal{P}}_d^1 \eqdef \{P_{r,s} \mid r,s  = 0, \dots, d-1 \}$, and transition probabilities $\mathcal{W}^{\#}(P_{r,s}\mid P_{t,u}) = a_{v,w}$, where $ v= r + t \: (\text{mod } d)  $ and $w= s + u \:(\text{mod } d) $. Consider now the channel combining and splitting procedure on $\mathcal{W}$, where $C \in \mathcal{C}_d^2$ is used to combine the two copies of $\mathcal{W}$. Let $\Gamma_C:\bar{\mathcal{P}}_d^1 \otimes \bar{\mathcal{P}}_d^1 \mapsto \bar{\mathcal{P}}_d^1 \otimes \bar{\mathcal{P}}_d^1 $ be the permutation induced by the conjugate action of $C$. We may define a channel combining and splitting procedure on the  classical $\mathcal{W}^{\#}$, using $\Gamma_C$ to combine the two copies of $\mathcal{W}^{\#}$. Similarly to~\cite{DGMS19}, we may  prove (but the proof is omitted here) that the Pauli qudit channel $\mathcal{W}$ and its classical counterpart $\mathcal{W}^{\#}$ {\em polarize simultaneously}, in the sense of~\cite[Proposition $20$ and Corollary $21$]{DGMS19}, under their respective channel combining and splitting procedure. As a consequence, to a quantum polar code on the Pauli qudit channel $\mathcal{W}$, we may associate a classical polar code on $\mathcal{W}^{\#}$, then exploit classical polar decoding in order to decode Pauli errors, as explained below (see also~\cite[Section~6]{DGMS19}). Let $\mathbf{P}$ denote  the unitary  corresponding to a quantum polar code of length $N$ qudits (see also \cite[Section~5]{DGMS19}), and $\mathbf{P}^{\#}$ the  linear map corresponding to the classical polar code. To perform decoding, we first apply $\mathbf{P}^\dagger$ on the $N$-qudit channel output, that is,  
 the encoded quantum state corrupted by some Pauli error, say $E\in (\bar{\mathcal{P}}_d^1)^{\otimes N}$ (we may omit phase factors). Hence, applying $\mathbf{P}^\dagger$ brings it back to the original (un-encoded) state, which is however corrupted by a Pauli error $E'\in (\bar{\mathcal{P}}_d^1)^{\otimes N}$, such that $\mathbf{P}^{\#}(E') = E$. We are now in position to decode $E'$, provided that we have been given the errors corresponding to the noisy virtual channels. We know that the inputs to the noisy channels are halves of preshared EPR pairs. Hence, we may perform projective measurements on the preshared EPR pairs, with respect to the generalized Bell basis $\{ I \otimes P_{r,s} \ket{\Phi_{AA'}}| P_{r,s} \in \bar{\mathcal{P}}_d^1  \}$, which give us the errors, {\em i.e.}, the $E'$ components, on the noisy virtual channels, as desired. Finally, we may decode the classical polar code to determine $E'$, and subsequently apply $E'^\dagger$ to return the system to the original quantum state.

\section{Conclusion and perspectives}
The goal of this work has been to generalize the purely quantum polarization  construction to higher dimensional quantum systems.  We have introduced the necessary definitions and worked out the proof of quantum polarization, assuming the channel combining unitary is randomized over  (1)  an unitary 2-design, (2) the two-qudit Clifford group, or (3) a smaller subset of two-qudit Cliffords.  Using  Clifford channel combining unitaries is important, as we showed it allows reducing the decoding problem to a classical polar code decoding, for qudit Pauli channels.  However, we note that the reliability of the classical polar code decoding also depends on the speed of polarization~\cite{arikan09}. We believe that fast polarization properties can also be generalized to the qudit case, although we leave this here as an open question.

\section*{Acknowledgements}

This research was supported in part by  the ``Investissements d’avenir'' (ANR-15-IDEX-02) program of the French National Research Agency. Ashutosh Goswami acknowledges the European Union's Horizon 2020 research and innovation programme, under the Marie Skłodowska Curie grant agreement No 754303.

%% %%%%%%%%%%%%%%%%%%%%%%%%%%%%%%%%%%%%% %%
%%      APPENDIX STARTS HERE             %%
%% %%%%%%%%%%%%%%%%%%%%%%%%%%%%%%%%%%%%% %%

\appendix

\section{Proof of Lemma~\ref{lem:twirling-Pd}}
\label{appx:proof:lem:twirling-Pd}

Recall that $\bar{\mathcal{P}}_d^2 = \left\{ P_{r,s} \otimes P_{r',s'} \mid  r,s,r',s' = 0,\dots, d-1 \right\}$ is the subset of two-qudit Pauli, without phase factors. Hence, 
twirling of $\mathcal{W}_2$ with respect to $\bar{\mathcal{P}}_d^2$ gives
\begin{equation} 
\mathcal{W}_2'(\rho) = \frac{1}{d^4}\sum_{r,s,r',s'} ( P_{r,s}^\dagger \otimes P_{r',s'}^\dagger) A \left( P_{r,s} \otimes P_{r',s'}\right) \rho ( P_{r,s}^\dagger \otimes P_{r',s'}^\dagger) B \left( P_{r,s} \otimes P_{r',s'}\right)
\end{equation}
\noindent Since $\bar{\mathcal{P}}_d^2$ forms an operator basis, we may write 
\begin{align}
A &= \sum_{r,s,r',s'} \alpha(r,s,r',s')  P_{r,s} \otimes P_{r',s'}, \label{eq:A-on-Pauli-basis}\\
B &= \sum_{r,s,r',s'} \beta(r,s,r',s')  P_{r,s} \otimes P_{r',s'} \label{eq:B-on-Pauli-basis}
\end{align}
Substituting $A$ and $B$ in the above equation, we get
\begin{align}
 \mathcal{W}_2'(\rho) &= \frac{1}{d^4} \sum_{t,u,t',u'} \: \sum_{v,w,v',w'} \alpha(t,u,t',u') \beta(v,w,v',w') \kappa, \label{eq:twirl_Pd} \\
 \text{where } \kappa &\eqdef \sum_{r,r',s,s'} (P_{r,s}^\dagger P_{t,u}P_{r,s}) \otimes (P_{r',s'}^\dagger P_{t',u'}P_{r',s'}) \rho \break (P_{r,s}^\dagger P_{v,w}P_{r,s}) \otimes (P_{r',s'}^\dagger P_{v',w'}P_{r',s'}). 
\end{align}
From~(\ref{eq:commutation-general}), we have that $P_{t,u}P_{r,s} = \omega^{-ru + st} P_{r,s} P_{t,u}$. Then, we may write
\begin{align}
\kappa &= k (P_{t,u} \otimes P_{t',u'}) \rho (P_{v,w} \otimes P_{v',w'})\\
\text{with } k &\eqdef \sum_{r,s} \omega^{-r(u + w) + s(v + t)} \sum_{r',s'} \omega^{-r'(u' + w') + s'(v' + t')}.
\end{align}
 When $u + w = v + t = 0\ (\text{mod } d)$, we have 
$\sum_{r,s} \omega^{-r(u + w) + s(v + t)} = d^2$. When either $u + v \neq 0\ (\text{mod } d)$ or $t + w \neq 0\ (\text{mod } d)$, we have  $\sum_{r,s} \omega^{-r(u + w) + s(v + t)} = \frac{(\omega^{-d} - 1)(\omega^{d} - 1)}{(\omega^{-1} - 1)(\omega - 1)} = 0$. Therefore,
\begin{equation}
k =
  \begin{cases}
    d^4, & \text{when } u + w = v + t = u' + w'= v' + t' = 0\ (\text{mod } d) \\
    0 , & \text{otherwise }
  \end{cases} 
\end{equation}
The condition $u + w = v + t = 0\ (\text{mod } d)$ implies that $P_{t,u}P_{v,w} = X^t Z^u X^v Z^w = \omega^{-uv} I$. Using $t = -v \ (\text{mod } d)$, we have that 
$P_{v,w} = \omega^{tu} P_{t,u}^\dagger$. Plugging $\kappa$ into~(\ref{eq:twirl_Pd}), we get
\begin{align}
 \mathcal{W}_2'(\rho) &= \sum_{t,u,t',u'}
\gamma_{t,u,t',u'} (P_{t,u} \otimes P_{t',u'}) \rho (P_{t,u}^\dagger \otimes P_{t',u'}^\dagger),\\
\text{where } \gamma_{t,u,t',u'} &\eqdef \omega^{tu + t'u' } \alpha(t,u,t',u') \beta(-t,-u,-t',-u'). \label{def:gamma-coefs}
\end{align}
Hence,  $\mathcal{W}_2'$ is a qudit Pauli channel, as desired.  
\hfill \qedsymbol

\section{Proof of Lemma~\ref{lem:twirling-Cd}}
\label{appx:proof:lem:twirling-Cd}

Recall that  $\bar{\mathcal{C}}_d^2 \subset  \mathcal{C}_d^2$ is a subset containing one representative for each  equivalence class in the quotient group $\mathcal{C}_d^2 / \mathcal{P}_d^2$. Twirling of $ \mathcal{W}_2'$ with respect to $\bar{\mathcal{C}}_d^2$ gives
\begin{equation}
\mathcal{W}_2''(\rho) = \sum_{t,u,t',u'} \gamma_{t,u,t'u'} \frac{1}{|\bar{\mathcal{C}}_d^2|}
\sum_{C \in \bar{\mathcal{C}}_d^2} C^\dagger (P_{t,u} \otimes P_{t',u'}) C \rho C^\dagger (P_{t,u}^\dagger \otimes P_{t',u'}^\dagger) C.
\end{equation}
We know that the conjugate action of the entire set $\bar{\mathcal{C}}_d^2$ maps any $P_{t,u} \otimes P_{t',u'} \neq I \otimes I$ to all $d^4 - 1$ two-qudit Paulis excluding $I \otimes I$, an equal number of times. In other words,  $P_{t,u} \otimes P_{t',u'} \neq I \otimes I$ gets mapped to a Pauli  $P_{r,s} \otimes P_{r',s'} \neq I \otimes I$, $\frac{|\bar{\mathcal{C}}_d^2|}{d^4 - 1}$ times. Further, $I \otimes I$ is always mapped to $I \otimes I$. Therefore, we have that
\begin{align} 
\mathcal{W}_2''(\rho) &= \gamma_{0,0,0,0} \rho + \frac{1}{d^4 - 1} \gamma'  \sum_{(r,s,r',s')\neq (0,0,0,0)}(P_{r,s} \otimes P_{r',s'}) \rho (P_{r,s}^\dagger \otimes P_{r',s'}^\dagger), \label{eq:twirl-CD}\\
\text{where } \gamma' &\eqdef \sum_{(t,u,t',u') \neq (0,0,0,0)} \gamma_{t,u,t',u'}.
\end{align}
Using the following three identities,  we can easily  transform~(\ref{eq:twirl-CD}) into the form of~(\ref{eq:twirl-C2}).

\begin{enumerate}
\item $\displaystyle \gamma_{0,0,0,0} = \frac{\text{Tr}(A) \text{Tr}(B)}{d^4}$.  

\item $\displaystyle  \sum_{t,u,t',u'} \gamma_{t,u,t',u'} = \frac{\text{Tr}(AB)}{d^2}$.

\item $\displaystyle \sum_{r,s,r',s'} (P_{r,s} \otimes P_{r',s'}) \rho (P_{r,s}^\dagger \otimes P_{r',s'}^\dagger) = d^2 I \otimes I$.
\end{enumerate}

%\medskip \noindent  {\bf\em Proof of the identities:}
%First recall that,
%\begin{align}
%A &= \sum_{r,s,r',s'} \alpha(r,s,r',s')  P_{r,s} \otimes P_{r',s'} \nonumber \\
%B &= \sum_{r,s,r',s'} \beta(r,s,r',s')  P_{r,s} \otimes P_{r',s'} \nonumber \\
%\gamma_{t,u,t',u'} &= \omega^{tu + t'u' } \alpha(t,u,t',u') \beta(-t,-u,-t',-u') \nonumber
%\end{align}

\medskip \noindent  {\em Proof of identity 1)}  We have that $\gamma_{0,0,0,0} = \alpha(0, 0, 0, 0) \beta(0, 0, 0, 0)$. Also, 
  $$ \text{Tr}(P_{r,s}) =
  \begin{cases}
    d, & \text{when } P_{r,s} = I \\
    0 , & \text{otherwise }
  \end{cases} 
  $$
Using~(\ref{eq:A-on-Pauli-basis}) and~(\ref{eq:B-on-Pauli-basis}), we get $ \text{Tr}(A) = \alpha(0, 0, 0, 0) d^2$ and $ \text{Tr}(B) = \beta(0, 0, 0, 0) d^2$. Hence, $\gamma_{0,0,0,0} = \frac{\text{Tr}(A) \text{Tr}(B)}{d^4}$. 

\medskip \noindent  {\em Proof of identity 2)} We have,
\begin{align}
\text{Tr}(AB) &= \sum_{t,u,t',u'} \: \sum_{v,w,v',w'} \alpha(t,u,t',u') \beta(v,w,v',w')  \text{Tr}(P_{t,u}P_{v,w}) \text{Tr}(P_{t',u'}P_{v',w'}) \nonumber \\
&= \sum_{t,u,t',u'}  d^2 \omega^{tu+t'u'} \alpha(t,u,t',u') \beta(-t,-u,-t',-u') \nonumber \\
&= d^2 \sum_{t,u,t',u'} \gamma_{t,u,t',u'}.  \nonumber 
\end{align}

\medskip \noindent  {\em Proof of identity 3)}
Let $\rho = \sum_{r,s,r',s'} \rho_{r,s,r',s'} P_{r,s} \otimes P_{r',s'}$. Since $\rho$ is a density matrix, we have $\rho_{0,0,0,0} = \frac{\tr(\rho)}{d^2} = \frac{1}{d^2}$. Hence, 
\begin{align*}
\sum_{r,s,r',s'} (P_{r,s} \otimes P_{r',s'}) \rho (P_{r,s}^\dagger \otimes P_{r',s'}^\dagger) &= \sum_{r,s,r',s'}  \: \sum_{t,u,t',u'} \rho_{t,u,t',u'} (P_{r,s}P_{t,u}P_{r,s}^\dagger) \otimes (P_{r',s'}P_{t',u'}P_{r',s'}^\dagger) \\
& =  \sum_{t,u,t',u'} \rho_{t,u,t',u'} \left(  \sum_{r,s,r',s'} \omega^{-st + ru} \omega^{-s't' + r'u'} \right) P_{t,u} \otimes P_{t',u'} \\
& = d^4 \rho_{0,0,0,0} I \otimes I \\
& = d^2  I \otimes I.  
\end{align*}
We get~(\ref{eq:twirl-C2}) from~(\ref{eq:twirl-CD}) by using the above identities, while also substituting the notation $\ident$ for the identity matrix $I$, as it denotes a quantum state here. \hfill \qedsymbol

\section{Proof of Lemma~\ref{lem:number-Clifford}}
\label{appx:proof:lem:number-Clifford}

 Consider  the one-qudit Clifford group $\mathcal{C}_d^1$. We count first the permutations generated by $\mathcal{C}_d^1$ on $\bar{\mathcal{P}}_d^1 \eqdef \{  P_{r,s}| r, s  = 0, \dots, d-1\}$, and later we will accommodate the phase factors. Any Clifford $C \in \mathcal{C}_d^1$ is uniquely determined by its conjugate action on the generators of the Pauli group, $X$ and $Z$. Suppose that $C$ maps $X \mapsto P_{r,s}$ and $Z \mapsto P_{t,u}$ via its conjugate action, where $P_{r,s}, P_{t,u} \neq I$. On the one hand, since commutation relations  are preserved under unitary conjugation, $P_{r,s}$ and $P_{t,u}$ must satisfy $P_{r,s} P_{t,u} = \omega P_{t,u} P_{r,s}$. On the other hand, from~(\ref{eq:commutation-general}), we have that $P_{r,s} P_{t,u} = \omega^{ru-st} P_{t,u} P_{r,s}$. Therefore, $r, u,s,t$ must be such that $ru - st = 1\ (\text{mod d})$. We fix $r,s$ and solve for $t,u$. Since $P_{r,s} \neq I$, it follows that either $r$ or $s$ is non-zero. Without loss of generality, we may assume that $r \neq 0$. Since $d$ is a prime number, $r$ is invertible under multiplication modulo $d$. Therefore, for any $t \in \{0,\dots, d-1\}$, there exists a unique $u := r^{-1}(1+st)\ (\text{mod } d)$, satisfying $ru - st = 1$. Hence, there are exactly $d$ choices for the $t, u$ pair. Since we have $d^2-1$ choices for the $r,s$ pair, it follows that there are $d(d^2-1)$ pairs of Paulis, $P_{r,s}$ and $P_{t,u}$, such that $P_{r,s}P_{t,u} = \omega P_{t,u}P_{r,s}$. Taking into account the phase factors, $\omega^\lambda, \lambda \in \{0,\dots,d-1\}$, it follows that $\mathcal{C}_{d}^1$ has $d^3(d^2-1)$ elements.

\medskip We now count the number of elements in $\mathcal{C}_d^2$. The two-qudit Pauli group $\mathcal{P}_d^2$ is generated by a set of four Paulis $I \otimes X, I \otimes Z, X \otimes I$ and $Z \otimes I$, and any Clifford $C \in \mathcal{C}_d^2$ is uniquely determined by its conjugate action on these four generators. The commutation relations between the four generators are illustrated in Fig.~\ref{fig:commute}.
\begin{figure}[h!]
     \centering
     \begin{tikzpicture}
     \coordinate[label=below:$I \otimes Z$] (A) at (0,0);
     \coordinate[label=above:$I \otimes X$] (B) at (0,0.7);
     \coordinate[label=below:$Z \otimes I$](C) at (1.5,0);
     \coordinate[label=above:$X \otimes I$](D) at (1.5,0.7);
     \draw [thick] (A) -- node[above]{} (B);
     \draw [thick] (C) -- node[above]{} (D);
     \end{tikzpicture}
    \caption{Connected Paulis satisfy $AB = \omega BA$, with $A$ is the Pauli on the top row, and $B$ the Pauli on the bottom row. Paulis that are not connected commute.}
    \label{fig:commute}
\end{figure}
Consider a mapping $I \otimes X \mapsto A$, $I \otimes Z \mapsto B$, $X \otimes I \mapsto A'$,  $Z \otimes I \mapsto B'$, where $A, B, A', B' \in \bar{\mathcal{P}}_d^2$,  that preserves all the commutation relations between generators. Pauli $I \otimes X$ can be mapped to any two-qudit Pauli $A \neq I \otimes I$, so there are $d^4 - 1$ choices for $A$. It is not very difficult to see that for any $A \neq I \otimes I$ there are $d^3$ choices for $B$ such that $AB = \omega BA$. Further, there are $d(d^2 -1)$ pairs of two-qudit Paulis $A'$ and $B'$, which commute with both $A$ and $B$, and satisfy $A'B' = \omega B'A'$. Therefore, we have $d^4(d^4 - 1) (d^2 - 1) $ possible permutations on $\bar{\mathcal{P}}_d^2$, which satisfy all the commutation relations. Taking into account the phase factors, it follows that $\mathcal{C}_{d}^2$ has $d^8(d^4 - 1) (d^2 - 1) $ elements. \hfill \qedsymbol

\section{Example of left coset fixed by the swap gate}
\label{appx:example-left-coset}

We consider $d = 5$. Let  $C_1 =I$ be the identity, and $C_2' \in \mathcal{C}_d^1$ be such that it maps $X \mapsto X^4 $ and $Z \mapsto Z^4$, via conjugation. Since $X^4Z^4 = \omega Z^4 X^4$, $C_2'$ is indeed a one-qudit Clifford. We define $C_2 = C_2'X^2Z^2 $. Further, let $C \in \mathcal{C}_d^2$, such that its conjugate action generates the following permutation on the generators of $\mathcal{P}_d^2$,
\begin{align}
I \otimes X &\mapsto X^4Z \otimes XZ^4, \nonumber \\
I \otimes Z &\mapsto XZ \otimes X^4Z^4, \nonumber \\
X \otimes I &\mapsto X^4Z \otimes X^4Z, \nonumber \\
Z \otimes I &\mapsto XZ \otimes XZ. \nonumber
\end{align}
Using~(\ref{eq:commutation-general}), it is easily seen that the above permutation preserves all the commutation relations between the generators. Now, the conjugate actions of $SC$ and $C(C_1 \otimes C_2)$ generate the same permutation on $\mathcal{P}_d^2$. Therefore, $SC = C(C_1 \otimes C_2)$.

\printbibliography
%\bibliographystyle{IEEEtran}
%\bibliography{isit}% bibliography style %
\typeout{get arXiv to do 4 passes: Label(s) may have changed. Rerun}
\end{document}